\documentclass[sn-basic]{sn-jnl}

\usepackage{graphicx}%
\usepackage{multirow}%
\usepackage{amsmath,amssymb,amsfonts}%
\usepackage{amsthm}%
\usepackage{mathrsfs}%
\usepackage[title]{appendix}%
\usepackage{xcolor}%
\usepackage{textcomp}%
\usepackage{manyfoot}%
\usepackage{booktabs}%
\usepackage{algorithm}%
\usepackage{algorithmicx}%
\usepackage{algpseudocode}%
\usepackage{listings}%

\usepackage{tikz}
\usetikzlibrary{arrows}
\usetikzlibrary{plothandlers}
\usetikzlibrary{decorations,decorations.pathmorphing,decorations.markings}

\theoremstyle{thmstyleone}%
\newtheorem{theorem}{Theorem}
\newtheorem{proposition}[theorem]{Proposition}
\newtheorem{lemma}[theorem]{Lemma}%

\theoremstyle{thmstyletwo}%
\newtheorem{example}{Example}%

\theoremstyle{thmstylethree}%
\newtheorem{definition}{Definition}%

\begin{document}
\renewcommand{\a}{\alpha}
\renewcommand{\b}{\beta}
\newcommand{\g}{\gamma}
\renewcommand{\d}{\delta}
\newcommand{\e}{\varepsilon}
\newcommand\z{\zeta}
\renewcommand{\th}{\vartheta}
\renewcommand{\k}{\kappa}
\newcommand{\n}{\nu}
\renewcommand{\l}{\lambda}
\newcommand{\no}{\omega}
\def\oo{\omega}
\renewcommand{\r}{\rho}
\newcommand{\s}{\sigma} 
\newcommand{\ph}{\varphi}
\renewcommand{\t}{\tau}
\newcommand{\Oo}{\Omega}
\newcommand{\Om}{\Omega}
\def\QP{{\bf QP}}

\newcommand\mutilde{\widetilde\mu}
\def\mtilde{\widetilde m}
\newcommand\nutilde{\widetilde\nu}
\newcommand\Mtilde{\widetilde M}
\newcommand\Ntilde{\widetilde N}
\def\Mctilde{\widetilde{\cal M}}
\def\Deltatilde{\widetilde{\Delta}}
\def\Gammatilde{\widetilde{\Gamma}}
\def\Sigmatilde{\widetilde{\Sigma}}
\def\Deltahat{\widehat\Delta}
\def\Mhat{\widehat M}
\def\Mchat{\widehat{\cal M}}
\def\Nhat{\widehat N}
\def\Sigmahat{\widehat\Sigma}
\def\Gammahat{\widehat\Gamma}
\def\muhat{\widehat\mu}
\def\nuhat{\widehat\nu}
\def\Rbar{{\overline\RR}}
\def\toexp{\mathop{\exp}^{\longrightarrow}}
\def\Rp{{\RR\setminus\{0\}}}
\def\intRn{\int_{\RR\setminus\{0\}}}
\def\Qphi{Q_\ph}
\def\eiox{e^{i\oo x}}
\def\ant#1#2{\{#1[\one],#2\}}
\def\Pfi#1{(#1-\ph(#1)\cdot\one)}
\def\ket#1{|#1\rangle}
\def\bra#1{\langle#1|}
\def\ketbra#1{|#1\rangle\langle#1|}
\def\supp{{\rm supp}}

\def\A{{\cal A}}
\def\B{{\cal B}}
\def\C{{\cal C}}
\def\D{{\cal D}}
\def\E{{\cal E}}
\def\F{{\cal F}}
\def\G{{\cal G}}
\def\H{{\cal H}}
\def\I{{\cal I}}
\def\K{{\cal K}}
\def\L{{\cal L}}
\def\M{{\cal M}}
\def\S{{\cal S}}

\def\Cbar{\overline{\cal C}}
\def\Cb{C_b}
\def\Lj{L_{\rm jump}}
\def\Ltilde{\tilde{L}}
\def\CP{{\rm CP}}
\def\norm#1{\|#1\|}
\def\Norm#1{\left\|#1\right\|}
\def\norms#1{\|#1\|^2}
\def\Norms#1{\left\|#1\right\|^2}
\def\normfour#1{\|#1\|_4}
\def\trinorm#1{|||#1|||}

\def\sod#1{\sum_{#1=1}^d}
\def\sok#1{\sum_{#1=1}^k}
\def\sokn#1{\sum_{#1=1}^{k_n}}
\def\sol#1{\sum_{#1=1}^l}
\def\som#1{\sum_{#1=1}^m}
\def\son#1{\sum_{#1=1}^n}
\def\sn#1{\sum_{#1=0}^{n-1}}
\def\sm#1{\sum_{#1=0}^{m-1}}
\def\szk#1{\sum_{#1=0}^k}
\def\szi#1{\sum_{#1=0}^\infty}
\def\soi#1{\sum_{#1=1}^\infty}
\def\li#1{\lim_{#1\to\infty}}
\def\lid#1{\lim_{#1\downarrow0}}
\def\szn#1{\sum_{#1=0}^n}
\def\intR{\int_{-\infty}^\infty}
\def\intRp{\int_{\RR\setminus\{0\}}}
\def\wlim{\qopname\relax m{w{-}lim}}
\def\ddt{\frac d {dt}}
\def\Mdop{\A^{\rm op}}
\def\dth{{\textstyle\frac 1 d}}
\def\Ld{L^{\textnormal{diff}}}
\def\Lj{L^{\textnormal{jump}}}
\def\Lz{L^{\textnormal{zero}}}
\def\half{{\textstyle{\frac 12}}}
\def\third{{\textstyle{\frac 13}}}
\def\twothirds{{\textstyle{\frac 23}}}
\def\quarter{{\textstyle{\frac 14}}}

\def\Dom{{\rm Dom}}
\def\id{{\rm Id}}
\def\Ad{{\rm Ad}}
\def\Re{\mathop{\rm{Re}}}
\def\Im{\mathop{\rm{Im}}}
\def\tr{{\mathrm{tr}}}

\def\EE{{\mathbb E}}
\def\NN{{\mathbb N}}
\def\ZZ{{\mathbb Z}}
\def\RR{{\mathbb R}}
\def\CC{{\mathbb C}}
\def\tuple#1#2{#1_1,#1_2,\ldots,#1_{#2}}
\def\one{{\mathchoice {\rm 1\mskip-4mu l} {\rm 1\mskip-4mu l}{\rm 1\mskip-4.5mu l} {\rm 1\mskip-5mu l}}}
\def\to{\rightarrow}
\def\tto{\longrightarrow}
\def\ten{\otimes}
\def\inp#1#2{\langle#1,#2\rangle}
\def\Inp#1#2{\left\langle#1,#2\right\rangle}
\def\implies{\Longrightarrow}
\def\Implies{\quad\Longrightarrow\quad}  
\def\and{\quad\hbox{and}\quad} 
\def\an#1{\quad\hbox{and #1}\quad} 
\def\twovector#1#2{\begin{pmatrix}#1\\#2\end{pmatrix}}
\def\twomatrix#1#2#3#4{\begin{pmatrix}#1&#2\\#3&#4\end{pmatrix}}
\def\twosvector#1#2{\left(\begin{smallmatrix}#1\\#2\end{smallmatrix}\right)}
\def\twosmatrix#1#2#3#4{\left(\begin{smallmatrix}#1&#2\\#3&#4\end{smallmatrix}\right)}
\def\set#1#2{\{#1|\;#2\;\}}
\def\Set#1#2{\bigg\{#1\bigg|\;#2\;\bigg\}}

\title[Article Title]{Copying Quantum States}


\author[1]{\fnm{Hans} \sur{Maassen}}\email{maassen@math.ru.nl}
\author[2]{\fnm{Burkhard} \sur{K\"ummerer}}\email{kuemmerer@mathematik.tu-darmstadt.de}


\affil[1]{\orgdiv{Mathematisch Instituut}, \orgname{Radboud Universiteit\par} \orgaddress{\street{Heyendaalseweg 135}, \city{Nijmegen}, \postcode{6525 AJ}, \country{Netherlands}}}

\affil[2]{\orgdiv{Fachbereich Mathematik}, \orgname{Technische Universit\"at Darmstadt}, \orgaddress{\street{Schlossgartenstraße 7}, \city{Darmstadt}, \postcode{64289}, \country{Germany}}}



\abstract{The no-broadcasting theorem in quantum information says that a set of states on a quantum system admits a common broadcasting (copying) operation
if and only if their density matrices belong to a commuting family.
We discuss and prove this theorem, as well as the closely related ``no-cloning theorem'' in the context of quantum probability theory,
i.e. in the category of (finite dimensional) C*-algebras with unital completely positive maps.}




\maketitle

\section{Physical Introduction}\label{Intro}

This issue is devoted to the memory of K.R. Parthasarathy.
Apart from his important contributions to analysis and probability theory,
he was a figurehead in the development of quantum probability theory and quantum stochastic calculus, e.g.:
\cite{HudsonPartha}, \cite{ParthaBook}, \cite{ParthaUncRel}, \cite{ParthaStates}.

\noindent
Quantum probability could be described as the study of C*-algebras and their completely positive maps,
to be used as a tool to extend probability theory to quantum mechanics \citep{BOOK}, 
\citep{KuemmererThesis}.
In this paper we apply this tool to discuss a fundamental theme of quantum information theory:
the (im)possibility of copying quantum states.

\noindent
The well-known ``no-cloning principle'',
which states that an unknown (pure) quantum state can not be copied, delineates quantum systems from the systems studied in classical information theory,
where unlimited copying is assumed possible.

\noindent
A `pure state cloner'  or `quantum amplifier'
is an imaginary device that takes as input a quantum system
in a state given by the vector $\psi$ in some Hilbert space $\H$,
and outputs $n\ge2$ identical quantum systems in the joint state $\psi^{\ten n}$.
This possibility was suggested by \cite{Herbert}, as a means to achieve superluminal communication.
It was discussed in the `Quantum Club' in Amsterdam by Dieks and Hoekzema, and
the existence of such a device was soon ruled out. The map $\psi\mapsto\psi^{\ten n}$ is not a linear map $\H\to\H^{\ten n}$ \citep{Dieks}.
Independently \cite{WoottersZureck} derived the result, using the word ``cloning'' in this context for the first time.

\noindent
The principle has been refined in several directions.
First it was realized \citep{Yuen} that, in the presence of certain foreknowledge, copying of pure states {\it is} possible: when the state vector is known to belong to a given orthonormal basis,
it can be cloned by  measuring in that basis, and then preparing a pair, say, of independent systems which are both in the pure state measured.
A second refinement is the ``no-broadcasting theorem'', extending the principle to mixed states.
It says that a mixed quantum state can be ``broadcast'' (copied) if and only if its density matrix is known to belong to a given commuting family.
This extension required ingeneous proofs, one by \cite{BarnumEtAl1} et al. using fidelities, and later one by \cite{Lindblad} based on matrix algebras,
which we shall treat here. Later a proof was given, based on convex structure only, by \cite{BarnumEtAl2} et al.

\noindent
In the standard textbook treatment, e.g. by \cite{NielsenChuang},
the no-cloning principle is presented in the orthonormal-set version mentioned above.
Due to a slight oversimplification it is made to look very simple (cf. our remark after Theorem \ref{ThmNoCloning}).
The no-broadcasting version at first sight would seem an easy consequence:
if the density matrices are diagonal in an orthonormal basis, then copying is possible;
otherwise there is a positive probability for the states to be non-orthogonal, so that no-cloning forbids them to be copied.

\noindent
To many researchers it therefore comes as a surprise that the proofs of the no-broadcasting theorem mentioned above
are quite nontrivial.

\noindent
In this paper we hope to shed some light on this matter.
Briefly, it is the ``only if''-direction which is difficult.
It must be shown that no quantum operation, say some ``quantum amplifier'' \citep{Yuen},
exists that can double up a state without having to perform any measurement.
Moreover the argument using probabilities of pure states cannot be used, since the decomposition of a mixed state into pure ones is far from unique.

\noindent
We hope to relieve the tension between the simplicity of ``no-cloning'' and the complexity of ``no-broadcasting'' proofs,
existing in the literature, by giving a fair, more complicated, proof of the first (Theorem \ref{ThmNoCloning}),
and a hopefully somewhat more transparant proof of the second (Theorem \ref{ThmNoBroadcasting}). 

\noindent
We feel that the no-cloning principle is not just a consequence of the linearity of quantum mechanics,
but rather of the non-commuting character of quantum observables, in the context of (completely) positive operations.
In fact, classical systems can also be formulated in a (commutative) linear model, and copying them is well possible.

\noindent
This paper is structured as follows:
In Section \ref{Preliminaries}\ we define matrix algebras, states, and operations,
and formulate basic general results concerning these.

\noindent
In Section \ref{CpDefs} we define the two notions of copying: cloning and broadcasting,
showing that one is a pure state version of the other.

\noindent
In Section \ref{Results} we prove ``no universal copying'' (Theorem \ref{ThmUnivCop}), ``no-broadcasting'' (Theorem \ref{ThmNoBroadcasting}),
and good old ``no-cloning'' (Theorem \ref{ThmNoCloning}).

\section{Mathematical Preliminaries}\label{Preliminaries}
\subsection{Matrix Algebras and States} 
By a {\it matrix algebra} $\A$ we shall mean a subspace of the linear space $M_d$ of all complex $d\times d$-matrices for some $d\in\NN$,
which is closed under matrix multiplication $(a,b)\mapsto ab$ and the taking of adjoints $a\mapsto a^*$,
and which contains a unit element $\one_\A$ (often the unit matrix in $M_d$).

\noindent
A {\it state} $\ph$ on $\A$ is a linear functional $\ph:\A\to\CC$ which maps positive definite matrices to positive numbers,
and satisfies the normalization $\ph(\one_\A)=1$.
The set of all states on $\A$ is denoted by $\S(\A)$, its extremal points forming the set $\E(\A)$ of {\it pure states}.

\begin{example}[\bf Pure Quantum Systems]\label{ExPureSyst}
The matrix algebra $M_d$ has state space
   $$\S(M_d)=\set{\ph:a\mapsto\tr(\rho a)}{\rho\in M_d, \rho\ge0, \tr(\rho)=1}\;.$$
The pure states are given by
   $$\E(M_d)=\Set{\ph:a\mapsto\inp{\psi}{a\psi}=\tr\bigl(\ketbra\psi a\bigr)}{\psi\in\CC^d,\norm\psi=1}\;.$$
\end{example}

\begin{example}[\bf Finite Classical Systems]\label{ClassSyst}
Let $\D_d\subset M_d$ denote the space of diagonal $d\times d$-matrices.
It has state space
   $$\S(\D_d)=\Set{\ph:a\mapsto\sod i \pi_i a_{ii}}{\tuple\pi d\ge0, \sod i \pi_i=1}\;,$$
of which the following are pure:
   $$\E(\D_d)=\{\tuple\d d\}\;,\quad\hbox{where}\quad \d_i(a):=a_{ii}\;.$$
\end{example}
\noindent
We note that $\D_d$ is isomorphic to the algebra $\F=\F(\Omega_d)$ of complex functions on the finite set $\Omega_d:=\{1,\ldots,d\}$,
if we let $f\in\F$ correspond to the diagonal matrix
   $$\begin{pmatrix}f(1)&0&\ldots&0\cr
                     0&f(2)&\ldots&0\cr
                     \vdots&\vdots&\ddots&\vdots\cr
                     0&0&\ldots&f(d)\cr\end{pmatrix}\;.$$
States on $\F(\Omega_d)$ are given by probability distributions $\pi:\Omega_d\to[0,1]$, making $\Omega_d$ a probability space.
By a slight abuse of notation we shall often identify $\D_d$ with $\F(\Omega_d)$.
It is in this sense that the theory of (finite) probability spaces $(\Omega_d,\pi)$ generalizes to that of {\it quantum probability spaces $(\A,\ph)$} \citep{KuemmererGreifswald}.     

\subsection{Completely Positive Maps}\label{QuantOp}
\noindent
A linear map $T$ from a matrix algebra $\A$ to a matrix algebra $\B$ is called {\it positive} if it maps positive definite matrices to positive definite matrices.
It is called {\it completely positive} if for all $n\in\NN$ the map $T\ten\id_{M_n}: \A\ten M_n\to\B\ten M_n$ is positive.
$T$ is called {\it unital} if $T(\one_\A)=\one_\B$.

\noindent
Matrix algebras with unital completely positive maps form a category, to be called $\QP$.

\noindent
By writing $T:\A\to\B$ we shall always mean that $T$ is a completely positive map from $\A$ to $\B$, mapping $\one_\A$ to $\one_\B$.
Such maps describe {\it quantum operations}: actions that can be performed on a quantum system. 
We illustrate the operation by the diagram

\bigskip
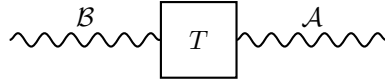
\begin{center}
\begin{tikzpicture}
\draw (1,0.3) node {$\B$};
\draw[thick,decorate,decoration={snake}] (0,0) -- (2,0);
\draw[thick] (2,-0.5) rectangle (3,0.5);
\draw (2.5,0) node {$T$};
\draw[thick,decorate,decoration={snake}] (3,0) -- (5,0) ;
\draw (4,0.3) node {$\A$};
\end{tikzpicture}
\end{center}

\noindent
to be read from right to left.
The map $T$ takes observables in $\A$ to observables in $\B$, and is said to be the operation, seen ``in the Heisenberg picture''.
The dual map $T^*:\S(\B)\to\S(\A)$, defined by
   $$T^*\ph(a):=\ph\bigl(T(a)\bigr)\;,$$
will be considered as the operation, seen in the ``Schr\"odinger picture'', and is denoted by the same diagram, now read from left to right.

\subsection{Physical Interpretation}\label{ProbInt}
A matrix algebra $\A\subset M_d$ may be interpreted as the {\it observable algebra} of a system,
represented on the Hilbert space $\H=\CC^d$.
Orthogonal projections $p=p^2=p^*\in\A$ stand for observational statements or ``events'' concerning the system.
If the system is in the state $\ph$, then $\ph(p)\in[0,1]$ is interpreted as the {\it probability} for $p$ to occur in an experiment.

\noindent
In the classical case of Example \ref{ClassSyst}, when $\A$ is commutative, then projections $p$ are of the form $1_A$ for some $A\subset\Omega$.
Then
   $$\ph(p)=\ph(1_A)=\sum_{i\in\Omega}\pi_i 1_A(i)=\sum_{i\in A} \pi_i\;,$$
i.e. the probability of the classical event $A$.

\noindent
A unital completely positive map $T:\A\to\B$, in the Schr\"odinger picture $T^*:\S(\B)\to\S(\A)$,
is interpreted as an action that can be performed on a system with algebra $\B$ in a state $\ph$, which yields a system with algebra $\A$ in the state $T^*\ph$.
Note that the time order is that of the Schr\"odinger picture.

\noindent
In order to prove ``no-go'' theorems such as the ``no-cloning'' and the ``no-broadcasting'' theorems it is necessary to have universal restrictions on what can be done to a system.
In the course of the 1960's and 1970's it was realized that the the category {\QP} is a good framework for the handling of open quantum systems.

\noindent
In the first place, the actions must be linear and positive, and this positivity must be robust against tensor products.
Probabilities must add up to 1. These are restrictions on ``what can be done'' to quantum systems, and they delineate the unital completely positive maps.
On the other hand, there seem to be no serious further restrictions: many unital completely positive maps are realized as quantum operations.
For instance, the following actions are unital and completely positive.

\begin{itemize}
\item[1]
{\it Free evolution of a closed system during some time}, described by automorphisms of the matrix algebra:
   $$a\mapsto u^*au\;,$$
where $u^*u=uu^*=\one$.
\item[2]
{\it Embedding} the system into a larger whole. For finite pure quantum systems, i.e. with $\A=M_d$, this must be coupling to an ancilla (auxiliary system),
according to Propositions 3.7 and 3.8 in \cite{KuemmererGreifswald}.
In the Sch\"odinger picture, this takes the form $\ph\mapsto\ph\ten\th$ for some state $\th$ on the ancilla.
In the Heisenberg picture it looks like
    $$\id\ten\th:\quad a\ten b\mapsto\th(b)a\;.$$
\item[3]
{\it Throwing away} or {\it ignoring} part of the system. In the Schr\"odinger picture this is represented by a {\it partial trace}.
If $\rho$ is the density matrix of a state on $M_d\ten M_n$:
   $$\rho\mapsto\tr_2(\rho):a\mapsto\tr\bigl(\rho(a\ten\one_{M_n})\bigr)\;.$$
 We prefer the simpler Heisenberg picture here:
    $$a\mapsto a\ten\one_{M_n}\;.$$
\end{itemize}

\noindent
Combinations of these are again quantum operations. By a variation on Stinespring's Theorem (Proposition \ref{PropStinespring}),
all unital completely positive maps can be written as a coupling to an ancilla, then a free evolution, followed by a restriction to a subsystem:
   $$T(a)=\id\ten\th\bigl(u^*(a\ten\one)u\bigr)\;.$$

\noindent
{\bf The Trivial *-Algebra.}
The matrix algebra $\CC=M_1$ is interpreted as `no system' or `no information'.
In diagrams we denote it by no line at all.
A state, when viewed as an operation $\ph$ from a system $\A$ to `no system' $\CC$, or, in de Schr\"odinger picture by a {\it preparation} of a quantum system `from nothing',
will be denoted by the diagram

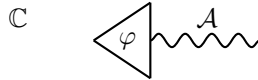
\begin{center}
\begin{tikzpicture}
\draw[thick] (0.75,-0.5) -- (0.75,0.5);
\draw[thick] (0.75,-0.5) -- (0,0);
\draw[thick] (0.75,0.5) -- (0,0);\draw(0.45,0)node{$\ph$};
\draw(1.5,0.3)node{$\A$};
\draw(-1,0.3)node{$\CC$};
\draw[thick,decorate,decoration={snake}] (0.75,0) -- (2.2,0);
\end{tikzpicture}
\end{center}

\noindent
There exists a single (completely) positive unital map from $\CC$ to $\A$, mapping $z\in\CC$ to $z\cdot\one\in\A$.
We interpret it physically as the operation of `destroying' or just `forgetting' the system $\A$.
We denote it by the diagram

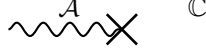
\begin{center}
\begin{tikzpicture}
\draw[thick,decorate,decoration={snake}] (2.5,0) -- (4,0);
\draw(3.3,0.3)node{$\A$};
\draw(5,0.3)node{$\CC$};
\draw[thick,line width=1pt](3.8,-0.2) -- (4.2,0.2);
\draw[thick,line width=1pt](3.8,0.2) -- (4.2,-0.2);
\end{tikzpicture}
\end{center}

\noindent
The mathematical fact that there exist many positive unital maps $\A\to\CC$, but only one $\CC\to\A$ corresponds neatly to the physical fact that
there are many ways to prepare a system, but only one to destroy it.

\subsection{Preliminary Results}\label{PrelRes}
Let us start with a basic fact about completely positive maps, which we shall give without proof. Consult, for example: \cite{NielsenChuang}, \cite{BOOK}, or \cite{Trieste}

\smallskip
\begin{proposition}[\bf Stinespring]\label{PropStinespring}
Let $T:M_m\to M_n$ be unital and completely positive.
Then there exists $k\in\NN$ and an isometry  $v:\CC^n\to\CC^m\ten\CC^k$ such that for all $a\in M_m$:
   $$T(a)=v^*(a\ten\one_{M_k})v\;.$$
\end{proposition}

\begin{proposition}[\bf Kadison-Schwarz]\label{PropCS}
Let $T$ be a unital completely positive map $\A\to\B$. Then we have, in the positive definite ordering of $\B$:
   $$T(a^*a)\ge T(a)^*T(a)\;.$$
\end{proposition}
\begin{proof}
Suppose that $\A\subset M_d$ and $\B\subset M_{d'}$.
For $a\in\A$, the matrix $\begin{pmatrix}a^*a&a\cr a^*&\one_\A\end{pmatrix}\in M_d\ten M_2$ is positive definite.
It then follows from the complete positivity of $T$ that for all $\psi\in\CC^{d'}$
   $$\inp{\psi}{\bigl(T(a^*a)-T(a)^*T(a)\bigr)\psi}=\begin{pmatrix}\psi,&-T(a)\psi\end{pmatrix}
                                                                                \begin{pmatrix}T(a^*a)&T(a^*)\cr T(a)&T(\one_\A)\end{pmatrix}
                                                                                \begin{pmatrix}\psi\cr -T(a)\psi\end{pmatrix}\ge0\;.$$
\end{proof}

\begin{proposition}[\bf Multiplicativity]\label{PropMult}
If $T(a^*a)=T(a)^*T(a)$ for some $a\in\A$, then we have for all $b\in\A$:
\begin{equation}\label{EqMult}
T(ba)=T(b)T(a)\;.
\end{equation}
\end{proposition}
We say that $a$ is {\it multiplicative} for $T$.

\begin{proof}
For every state $\ph$ on $\A$, define the semi-inner product $\inp\cdot\cdot_{T,\ph}$ on $\A$ by
   $$\inp a b_{T,\ph}:=\ph\bigl(T(a^*b)-T(a)^*T(b)\bigr)\;.$$
By Proposition \ref{PropCS} we then have for all $a\in\A$ that $\inp a a_{T,\ph}\ge0$.
By the Cauchy-Schwarz inequality for $\inp\cdot\cdot_{T,\ph}$ we have
   $$|\inp a b_{T,\ph}|^2\le|\inp a a_{T,\ph}|\cdot|\inp b b_{T,\ph}|\;.$$
Hence, if $\inp a a_{T,\ph}=\ph\bigl((Ta^*a)-T(a)^*T(a)\bigr)=0$, then for all $b\in\A$:
   $$\ph\bigl((T(ba)-T(b)T(a)\bigr)=\inp{b^*}a_{T,\ph}=0\;.$$
  And since this holds for all states $\ph$ on $\A$, we may conclude that (\ref{EqMult}) holds.
\end{proof}

\begin{definition}
A set $S$ of states on a matrix algebra $\A$ is called {\em faithful} if for all nonnegative definite $a\in\A$:
\begin{equation}\label{EqFaithful}
\biggl(\forall_{\ph\in S}: \ph(a)=0\biggr)\Implies a=0\;.
\end{equation}
\end{definition}

\noindent
We say that an observable $a\in\A$ is {\it invariant} under $T:\A\to\A$ if $T(a)=a$.
A state $\ph$ on $\A$ is {\it invariant} under $T$ if $T^*\ph=\ph\circ T=\ph$.

\smallskip\begin{proposition}[\bf Fixed Point Algebra]\label{PropFixAlgebra}
If $T:\A\to\A$ has a faithful set $S\subset\S(\A)$ of invariant states, then the set $\I(T)$ of all invariant observables forms a matrix subalgebra of $\A$.
\end{proposition}

\begin{proof}
If $a\in\I(T)$, then also $a^*\in\I(T)$, since $T(a^*)=T(a)^*$.
Moreover for all $\ph\in S$ we have
   $$0=\ph\circ T(a^*a)-\ph(a^*a)=\ph\bigl(T(a^*a)-T(a)^*T(a)\bigr)\;.$$
By Kadison-Schwarz (Proposition \ref{PropCS}) $T(a^*a)-a^*a=T(a^*a)-T(a)^*T(a)\ge0$, hence by faithfulness (\ref{EqFaithful})
it follows that $T(a^*a)=a^*a$, i.e., $a^*a\in\I(M)$.
Replacing $a$ by $a+b$, where $a$ and $b$ are both in $\I(T)$, and then by $a+ib$, and subtracting
(i.e. by ``polarization'') we obtain that $T(a^*b)=a^*b$: if $a$ and $b$ are invariant, then so is $a^*b$.
\end{proof}

\section{Definitions of Copying}\label{CpDefs}
A copying operation should map a system with observable algebra $\A$ in a state $\ph$ to a pair of systems,
described by the algebra $\A\ten\A$, each in the state $\ph$.
Hence in the Schr\" odinger picture we want our copier to be of the form
    $$C^*:\S(\A)\to\S(\A)\ten\S(\A)\;.$$
For the Heisenberg picture this means
    $$C:\A\ten\A\to\A\;.$$  
\begin{definition}[\bf Cloning]\label{DefCloner}
We say that $C:\A\ten\A\to\A$ {\em clones} the pure state $\ph$ on $\A$ if
   $$C^*\ph=\ph\ten\ph\;.$$
\end{definition}

\noindent
We only consider cloners for {\it pure} states $\ph$.
The reason for this restriction is explained easiest by considering the classical algebra $\F(\Omega)$ for the finite set $\Omega=\{A,B\}$,
where $A$ and $B$ may be documents.
Suppose we feed $A$ into the device with probability $\third$ and $B$ with probability $\twothirds$,
and suppose that it, with this probability distribution $(\third,\twothirds)$ as input,
yields as output the product distribution $\textstyle(\frac19,\frac29,\frac29,\frac49)$
on $\{(A,A),(A,B),(B,A),(B,B)\}$.
Then the copies are {\it independent}, and not {\it identical}.
This is not what a copying machine is supposed to do: it should map $(\third,\twothirds)$ to $(\third,0,0,\twothirds)$,
producing the pair $(A,A)$ with probability $\third$ and $(B,B)$ with probability $\twothirds$.

\noindent
In this example only for the extremal distributions $\d_A$ and $\d_B$ is this behaviour reasonable:
\begin{equation}\label{EqDeltas}
\d_A\mapsto\d_A\ten\d_A\quad\hbox{and}\quad\d_B\mapsto\d_B\ten\d_B\;.
\end{equation}

\noindent
Instead of Definition \ref{DefCloner},
for mixed states \cite{BarnumEtAl1}\ et al. proposed the following weaker definition,
asking the marginals of the output distribution to reflect the input state without requiring independence. 

\begin{definition}[\bf Broadcasting]\label{DefBroadcaster}
The map $C:\A\ten\A\to\A$ is said to {\em broadcast} the state $\ph$ on $\A$ if for all $a\in\A$:
\begin{equation}\label{EqBroadcast}
\ph\circ C(a\ten\one_\A)=\ph\circ C(\one_\A\ten a)=\ph(a)\;.
\end{equation}
\end{definition}
\noindent
The following diagram expresses this situation:

\smallskip
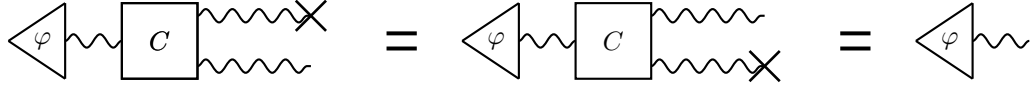
\begin{center}
\begin{tikzpicture}
\draw[thick,decorate,decoration={snake}] (0.75,0) -- (1.5,0);
\draw[thick] (1.5,-0.5) rectangle (2.5,0.5);\draw (2,0) node {$C$};
\draw[thick,decorate,decoration={snake}] (2.5,.333) -- (4,.333);
\draw[thick,decorate,decoration={snake}] (2.5,-0.333) -- (4,-0.333);
\draw[thick,line width=1pt](3.8,0.133) -- (4.2,0.533);
\draw[thick,line width=1pt](3.8,0.533) -- (4.2,0.133);
\draw[thick] (0.75,-0.5) -- (0.75,0.5);
\draw[thick] (0.75,-0.5) -- (0,0);
\draw[thick] (0.75,0.5) -- (0,0);\draw(0.45,0)node{$\ph$};
\draw[thick, line width=1.5pt](5,-0.1) -- (5.4,-0.1);  
\draw[thick,line width=1.5pt](5,0.1) -- (5.4,0.1);  
\draw[thick,decorate,decoration={snake}] (6.75,0) -- (7.5,0);
\draw[thick] (1.5,-0.5) rectangle (2.5,0.5);\draw (2,0) node {$C$};
\draw[thick,decorate,decoration={snake}] (8.5,.333) -- (10,.333);
\draw[thick,decorate,decoration={snake}] (8.5,-0.333) -- (10,-0.333);
\draw[thick] (7.5,-0.5) rectangle (8.5,0.5);\draw (8,0) node {$C$};
\draw[thick,line width=1pt](9.8,-0.533) -- (10.2,-0.133);
\draw[thick,line width=1pt](9.8,-0.133) -- (10.2,-0.533);
\draw[thick] (6.75,-0.5) -- (6.75,0.5);
\draw[thick] (6.75,-0.5) -- (6,0);
\draw[thick] (6.75,0.5) -- (6,0);\draw(6.45,0)node{$\ph$};
\draw[thick, line width=1.5pt](11,-0.1) -- (11.4,-0.1);  
\draw[thick,line width=1.5pt](11,0.1) -- (11.4,0.1);  
\draw[thick] (12.75,-0.5) -- (12.75,0.5);
\draw[thick] (12.75,-0.5) -- (12,0);
\draw[thick] (12.75,0.5) -- (12,0);\draw(12.45,0)node{$\ph$};
\draw[thick,decorate,decoration={snake}] (12.75,0) -- (13.5,0);
\end{tikzpicture}
\end{center}

\smallskip\noindent
Broadcasting the state $\ph$, and throwing away (or ignoring) one of the copies, leaves us with a system in the state $\ph$.

\noindent
Henceforth we shall use the word ``copying'' in the nontechnical sense, and the words ``cloning'' and ``broadcasting'' in the sense of definitions  \ref{DefCloner} and \ref{DefBroadcaster} above.

\noindent
Actually, the term ``cloning'' is superfluous, as the following lemma shows.
\smallskip
\begin{lemma}\label{LemPureCloner}
Let $\ph$ be a pure state on $\A$. Then $C:\A\ten\A\to\A$ clones $\ph$ if and only if it broadcasts $\ph$.
\end{lemma}

\begin{proof}
Clearly, if $C^*\ph=\ph\ten\ph$, then both marginals are equal to $\ph$:
   $$\ph\circ C(a\ten\one)=C^*\ph(a\ten\one)=\ph\ten\ph(a\ten\one)=\ph(a)\ph(\one)=\ph(a)\;,$$
and the same for $\one\ten a$.

\noindent
Conversely, suppose that the right marginal of $C^*\ph$ is equal to $\ph$ itself, as in (\ref{EqBroadcast}).
Let $p\in\A$ be a projection. Then we can write $\ph$ as a sum of two positive functionals:
   $$\ph(a)=C^*\ph(\one\ten a)=C^*\ph(p\ten a)+C^*((\one-p)\ten a)\;.$$
Since $\ph$ is pure, both must be multiples of $\ph$, in particular $C^*\ph(p\ten a)=\l\ph(a)$ for
some $\l\ge0$. Putting $a=\one$ we find that $\l=C^*\ph(p\ten\one)$,
and since the left marginal also equals $\ph$, we have $\l=\ph(p)$.
Thus we find
   $$C^*\ph(p\ten a)=\ph(p)\ph(a)\;.$$
As the projections $p$ span the algebra $\A$, we conclude that $C^*\ph=\ph\ten\ph$.
\end{proof}

\section{Results}\label{Results}
\subsection{Copying All States}
Let us first determine under what circumstances a single quantum operation can copy {\it all states} on a system.
The answer is: only if the system is classical:
\smallskip
\begin{theorem}[\bf No Universal Copier]\label{ThmUnivCop}
For a matrix algebra $\A$ and a unital completely positive map $C:\A\ten\A\to\A$ the following are equivalent:
\begin{itemize}
\item[\rm(a)]
$C$ clones all pure states on $\A$;
\item[\rm(b)]
$C$ broadcasts all states on $\A$;
\item[\rm(c)]
$\A$ is commutative and for all $a,b\in\A$ we have $C(a\ten b)=ab$.
\end{itemize}
\end{theorem}

\begin{proof}
(a)$\;\Longrightarrow\;$(b):
As noted in Lemma \ref{LemPureCloner}, the cloner
$C$ is also a copier for all pure states: $\ph\circ C(a\ten\one)=\ph\circ C(\one \ten a)=\ph(a)$.
Since any state is a convex combination of pure states, this equality
extends to all states $\ph$ on $\A$.

\noindent
(b)$\;\Longrightarrow\;$(c):
As the broadcasting equation (\ref{EqBroadcast}) holds for all $\ph$ on $\A$, we have, for all $a\in\A$:
   $$C(a\ten\one)=C(\one\ten a)=a\;.$$
It follows that
   $$C\bigl((a\ten\one)^*(a\ten\one)\bigr)=C(a^*a\ten\one)=a^*a=C(a\ten\one)^*C(a\ten\one)\;;$$
so  $a\ten\one$ is multiplicative for $C$ by Proposition \ref{PropMult}.
In the same way we see that for all $b\in\A$ the matrix $\one\ten b$ is multiplicative.
Therefore
\begin{eqnarray*}
ab&=&C(a\ten\one)C(\one\ten b)=C\bigl((a\ten\one)(\one\ten b)\bigr)=C(a\ten b)\cr
     &=&C\bigl((\one\ten b)(a\ten\one)\bigr)=C(\one\ten b)C(a\ten\one)=ba\;.
\end{eqnarray*}

(c)$\;\Longrightarrow\;$(a):
We may identify $\A$ with the algebra $\F(\Om)$ of all functions on some finite set $\Om$ (cf. Example \ref{ClassSyst}).
The pure states on this algebra are the point evaluations $\d_i:f\mapsto f(i)$, with $i\in\Om$.
These satisfy, for $f,g\in\F(\Om)$:
\begin{eqnarray*}
(C^*\d_i)(f\ten g)&=&\d_i\circ C(f\ten g)=\d_i(f\cdot g)=(f\cdot g)(i)\cr
   &=&f(i)g(i)=\d_i(f)\d_i(g)=(\d_i\ten\d_i)(f\ten g)\;.
\end{eqnarray*}
I.e.: $C^*\d_i=\d_i\ten\d_i$ for the pure states $\d_i$, $i\in\Omega$.
\end{proof}

\subsection{Classical Copiers.}\label{ClassCop}
Condition (c) above says that our system is classical, and that $C$ takes the form of the product operation.
Although this formulation is mathematically neat, physically it may look puzzling.
The following equivalent form may be more transparant.
If $\A=\F(\Om)$, then the map $C:f\ten g\to f\cdot g$ extends linearly to all $h\in\F(\Om)\ten\F(\Om)=\F(\Om\times\Om)$ as
   $$C(h)(i)=h(i,i)\;.$$
This is the algebraic form of the natural copying operation on $\Om$:
   $$i\mapsto(i,i)\;,$$
which takes a system in the `point state' $i$ to two systems, both in that point state.
In the Schr\"odinger picture the classical copier takes the form
\begin{equation}\label{EqCopSchr}
C^*\pi=\pi\circ C=\sum_{i\in\Om}\pi_i\cdot\d_i\ten\d_i\;.
\end{equation}
$C^*$ puts the probability distribution $\pi$ on $\Om$ onto the diagonal of $\Om\times\Om$.
A distribution $\pi$ on the points $i\in\Om$ is mapped to the same distribution on the pairs $(i,i)$.
(See our remarks after Definition \ref{DefBroadcaster}.)
We shall denote the copier of an abelian matrix algebra by the ``fork'' symbol

\smallskip
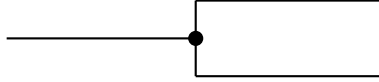
\begin{center}
\begin{tikzpicture}
\draw[thick] (0,0) -- (2.5,0);
\draw[thick] (2.5,0.5) -- (5,0.5) ;
\draw[thick] (2.5,-0.5) -- (5,-0.5) ;
\draw[thick] (2.5,0.5) -- (2.5,-0.5) ;
\fill[color=black] (2.5,0) circle [radius=0.1];
\end{tikzpicture}
\end{center}

\smallskip\noindent
Here, all the horizontal straight lines stand for the commutative matrix algebra $\F(\Om)$.

\noindent
Theorem \ref{ThmUnivCop}\ makes it clear that classical copiers do exist.

\noindent
They do their work by measuring the point-state $i\in\Om$, and producing a pair $(i,i)$ of them.
(Cf. (\ref{EqDeltas})).
We note that a copying machine does not {\em measure} the probability distribution in order to copy it.
A classical broadcaster does not distinguish between the probability distributions $(\third,\frac23)$ and $(\half,\half)$
on $\Om:=\{A,B\}$ in order to faithfully reproduce the distribution on the copies.
The probability distribution takes care of itself.

\noindent
This phenomenon makes broadcasting very different from mere cloning.


\subsection{Copying from a Given Set of States}\label{CopSetS}
\begin{theorem}[\bf No-Broadcasting Theorem]\label{ThmNoBroadcasting}
Let $S$ be a set of states on a matrix algebra $\A\subset M_d$.
Then the following are equivalent.
\begin{itemize}
\item[\rm(a)]
The states in $S$ have a common broadcaster $C:\A\ten\A\to\A$;
\item[\rm(b)]
All states in $S$ are convex combinations of certain states $\tuple\th k$ on $\A$, which have pairwise orthogonal support projections $\tuple p k\in\A$.
\item[\rm(c)]
The density matrices of the states in $S$ commute pairwise.
\end{itemize}
\end{theorem}

\noindent
The hard part of the proof is the implication (a)$\Longrightarrow$(b):
it basically says that the only way to broadcast a set of states is to perform a von Neumann measurement on the system  using certain projections $\tuple p k$,
and on the basis of the outcome to decide which state to output in duplicate.

\begin{proof}
(b)$\Longrightarrow$(c):
Since the states $\tuple \th k$ have orthogonal supports, their density matrices commute,
and so do convex combinations of these.

\smallskip\noindent
(c)$\Longrightarrow$(b):
Let $\C\subset\A$ denote the commutative matrix algebra generated by the density matrices $\set{\rho_\ph}{\ph\in S}$. 
After diagonalization of this algebra we divide the diagonal into (say $k$) blocks, inside which all the $\rho_\ph$ are constant.
Thus we obtain a division of $\one_\A$ into disjoint projections $\tuple p k$. By normalization we define states
$\tuple\th k$ with supports $\tuple p k$:
   $$\th_i(a):=\frac{\tr(p_ia)}{\tr p_i}\;.$$
Note that $\th_i(p_j)=\d_{ij}$.
Every state $\ph\in S$ is a convex combination of these states: $\ph=\sok i \l^\ph_i\th_i$.
In particular
$\ph(p_j)=\sok i\l^\ph_i\th_i(p_j)=\l^\ph_j$, so that
\begin{equation}\label{EqConvComb}
\ph=\sok i \ph(p_i)\th_i\;,
\end{equation}
with $\sok i \ph(p_i)=\ph(\one_\A)=1$.

\smallskip\noindent
(b)$\Longrightarrow$(a):
We may assume (\ref{EqConvComb}).
Define $C$ in the Schr\"odinger picture by
\begin{equation*}
C^*:\S(\A)\to\S(\A)\ten\S(\A):\quad \ph\mapsto\sok j \ph(p_j)\th_j\ten\th_j\;.
\end{equation*}
Then $C$ is a broadcaster of $S$, since for all $\ph\in S$, by (\ref{EqConvComb}),
   $$(C^*\ph)(a\ten\one)=(C^*\ph)(\one\ten a)=\sok j \ph(p_j)\th_j(a)=\ph(a)\;.$$
See Definition \ref{DefBroadcaster}; note that in the Heisenberg picture
\begin{equation}\label{EqBroad}
C(a\ten b)=\sok i\th_i(a)\th_i(b)p_i\;.
\end{equation}

\smallskip\noindent
(a)$\Longrightarrow$(b):
Suppose that a broadcaster $C:\A\ten\A\to\A$ for all the states $\ph\in S$ is given.
We may assume that $C$ is symmetric, otherwise we continue the discussion with the symmetrized broadcaster $a\ten b\mapsto\half\bigl(C(a\ten b)+C(b\ten a)\bigr)$.

\noindent
Let $s$ be the support of $S$, i.e., the smallest projection in $\A$ such that $\ph(s)=1$ for all $\ph\in S$.
Let $\A_s:=s\A s$, a matrix algebra with unit $\one_{\A_s}=s$.
By $S_s$ we shall denote the set of restrictions of the states $\ph\in S$ to $\A_s$.

\noindent
Define $C_s:\A_s\ten\A_s\to\A_s$ by:
     $$C_s(a\ten b):=sC(a\ten b)s\;.$$
Take $\psi\in S_s$, say $\psi(a)=\ph(a)$ for $a\in\A_s$, and some $\ph\in S$.
Since $\ph(s)=1$, the projection $s$ is multiplicative for $\ph$,
and since $\ph\circ C(\one\ten s)=1$, the matrix $\one\ten s$ is multiplicative for $\ph\circ C$.
Therefore, still for all $a\in\A_s$
\begin{eqnarray*} 
\psi\circ C_s(a\ten\one_{\A_s})&=&\ph\bigl(sC(a\ten s)s\bigr)=\ph(s)\ph\Bigl(C\bigl((a\ten\one)(\one\ten s)\bigr)\Bigr)\ph(s)\cr
   &=&\ph\circ C(a\ten\one)\cdot\ph\circ C(\one\ten s)=\ph(a)\ph(s)=\ph(a)=\psi(a)\;.
\end{eqnarray*}
I.e., $C_s$ broadcasts all states in $S_s$.

\noindent
Consider the marginal operation $M:\A_s\to\A_s:a\mapsto C_s(a\ten\one_{\A_s})$.
By assumption, $\psi\circ M=\psi$ for all $\psi\in S_s$.
As $S$ has support $s$, this set is faithful for $M$. By Proposition \ref{PropFixAlgebra} the invariant set $\I(M)$ is a matrix subalgebra of $\A_s$.

\noindent
It follows that, for $a\in\I(M)$ we have $a^*a\in\I(M)$, hence
   $$C_s\bigl((a\ten s)^*(a\ten s)\bigr)=C_s(a^*a\ten s)=M(a^*a)=a^*a=C_s(a\ten s)^*C_s(a\ten s)\;,$$
so $a\ten s$ is multiplicative for $C_s$. Therefore, for all $a,b\in\I(M)$,
\begin{equation}\label{EqImC}
C_s(a\ten b)=C_s(a\ten\one_{\A_s})\cdot C_s(\one_{\A_s}\ten b)=ab\in\I(M).
\end{equation}
We conclude that $C_s$ in restriction to $\I(M)\ten\I(M)$ maps to $\I(M)$ itself.
By definition of $\I(M)$ the map $C_s$ broadcasts all states on $\I(M)$.
But by Theorem \ref{ThmUnivCop} this is only possible if $\I(M)$ is commutative!

\noindent
Let $\tuple p k$ be the minimal projections in the commutative matrix algebra $\I(M)$.
Consider the ergodic projection $E:\A_s\to\I(M)$:
   $$E:=\li n\frac1n\sum_{j=0}^{n-1}M^j\;.$$
Since every element of $\I(M)$ is a linear combination of the projections $\tuple p k$,
we may write, for $a\in\A_s$,
\begin{equation}\label{EqErgProj}
E(a)=\sok i \a_i(a)p_i\;.
\end{equation}
for some functionals $\tuple\a k: \A_s\to\CC$.
As $E$ is positive, so are the $\a_i$, and as $E(\one_{A_s})=\one_{A_s}$, we find that $\a_i(\one_{\A_s})=1$.
Define $\tuple\th k$ by $\th_i(a):=\a_i(sas)$.

\noindent
We now have for $j=1,\ldots,k$:
   $$p_j=E(p_j)=\sok i\th_i(p_j)p_i\,,$$
and therefore $\th_i(p_j)=\d_{ij}$, i.e., the states have disjoint supports.
Moreover, for all $a\in\A$ and $\ph\in S$ we have by (\ref{EqErgProj}):
   $$\ph(a)=\ph(sas)=\ph\circ E(sas)=\sok i\a_i(sas)\ph(p_i)=\sok i\ph(p_i)\th_i(a)\;.$$
Since $\sok i p_i=s$ and $\ph(s)=1$,
the states $\ph\in S$ are convex combinations of the states $\tuple\th k$.
\end{proof}

\noindent{\bf Discussion of the No-Broadcasting Theorem.}
The von Neumann measurement using the projections $\tuple p k$ can be seen as a quantum operation
from a quantum system to a classical probability space given in the Heisenberg picture by
   $$N:\F(\Omega_k)\to\A:\quad f\mapsto\sok i f(i)p_i\;.$$
The subsequent preparation of a state $\th_i$ is also a quantum operation, $\Theta$ say:
   $$\Theta:\A\to\F(\Omega_k):\quad \theta(a)(i):=\th_i(a)\;.$$
Equation (\ref{EqConvComb}) can then be written as
\begin{equation}\label{EqKluunOp}
\ph\circ N\circ\Theta=\ph\;,
\end{equation}
in a picture:

\smallskip
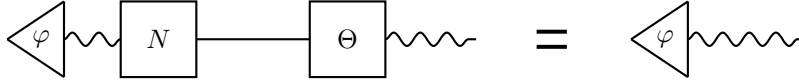
\begin{center}
\begin{tikzpicture}
\draw[thick,decorate,decoration={snake}] (0.75,0) -- (1.5,0);
\draw[thick] (1.5,-0.5) rectangle (2.5,0.5);\draw (2,0) node {$N$};
\draw[thick] (2.5,0) -- (4,0);
\draw[thick] (0.75,-0.5) -- (0.75,0.5);
\draw[thick] (0.75,-0.5) -- (0,0);
\draw[thick] (0.75,0.5) -- (0,0);\draw(0.45,0)node{$\ph$};
\draw[thick] (4,-0.5) rectangle (5,0.5);\draw (4.5,0) node {$\Theta$};
\draw[thick,decorate,decoration={snake}] (5,0) -- (6.2,0);
\draw[thick, line width=1.5pt](7,-0.1) -- (7.4,-0.1);  
\draw[thick,line width=1.5pt](7,0.1) -- (7.4,0.1);  
\draw[thick] (9,-0.5) -- (9,0.5);
\draw[thick] (9,-0.5) -- (8.25,0);
\draw[thick] (9,0.5) -- (8.25,0);\draw(8.7,0)node{$\ph$};
\draw[thick,decorate,decoration={snake}] (9,0) -- (10.5,0);
\end{tikzpicture}
\end{center}

\smallskip\noindent
It says that $\ph$ can pass through the ``entanglement breaking channel'' $N\circ\Theta$.

\noindent
Equation (\ref{EqBroad}) can be written as
\begin{equation}\label{EqBroadOp}
C=N\circ C_{\F(\Omega_k)}\circ(\Theta\ten\Theta)\;,
\end{equation}
in a picture:

\smallskip
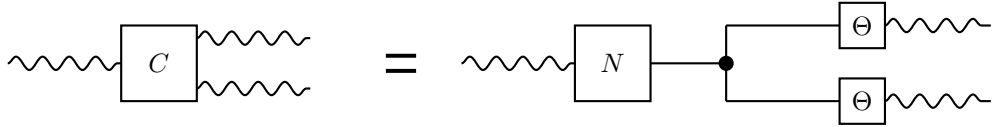
\begin{center}
\begin{tikzpicture}
\draw[thick,decorate,decoration={snake}] (0,0) -- (1.5,0);
\draw[thick] (1.5,-0.5) rectangle (2.5,0.5);\draw (2,0) node {$C$};
\draw[thick,decorate,decoration={snake}] (2.5,.333) -- (4,.333);
\draw[thick,decorate,decoration={snake}] (2.5,-0.333) -- (4,-0.333);
\draw[thick, line width=1.5pt](5,-0.1) -- (5.4,-0.1);  
\draw[thick,line width=1.5pt](5,0.1) -- (5.4,0.1);  
\draw[thick,decorate,decoration={snake}] (6,0) -- (7.5,0);
\draw[thick] (8.5,0) -- (9.5,0);
\draw[thick] (7.5,-0.5) rectangle (8.5,0.5);\draw (8,0) node {$N$};
\draw[thick] (9.5,0.5) -- (11,0.5) ;
\draw[thick] (9.5,-0.5) -- (11,-0.5) ;
\draw[thick] (9.5,0.5) -- (9.5,-0.5) ;
\fill[color=black] (9.5,0) circle [radius=0.1];
\draw[thick] (11,0.2) rectangle (11.6,0.8);\draw (11.3,0.5) node {$\Theta$};
\draw[thick] (11,-0.2) rectangle (11.6,-0.8);\draw (11.3,-0.5) node {$\Theta$};
\draw[thick,decorate,decoration={snake}] (11.6,0.5) -- (13,0.5);
\draw[thick,decorate,decoration={snake}] (11.6,-0.5) -- (13,-0.5);
\end{tikzpicture}
\end{center}

\smallskip\noindent
where the straight lines stand for the commutative algebra $\F(\Omega_k)$ and the `fork' denotes its classical broadcaster described in Section \ref{ClassCop}.
\noindent
Theorem \ref{ThmNoBroadcasting} says that the states in $S$ can be broadcast by the same quantum operation $C$ if and only if
they can all pass through a channel of the form (\ref{EqKluunOp}) (first picture).
If so, then a broadcaster for them is given by
equation (\ref{EqBroadOp}) and the second picture.

\subsection{Copying Pure States}
According to our No-Broadcasting Theorem \ref{ThmNoBroadcasting}, 
copying quantum states becomes possible as soon as the state to be reproduced is known
to lie in a set $S$ of states whose density matrices commute with each other.
For pure states $\psi_1$ and $\psi_2$ this means that $\ketbra{\psi_1}\cdot\ketbra{\psi_2}=\ketbra{\psi_2}\cdot\ketbra{\psi_1}$, i.e.
$\psi_1\parallel\psi_2$ or $\psi_1\perp\psi_2$.
The standard textbook formulation of this fact is called the ``No-Cloning Theorem''.
Since the pure state case figures in all textbooks, and has a much shorter proof \citep{Yuen},
we treat it separately below.

\begin{theorem}[\bf No-Cloning]\label{ThmNoCloning}
Let $S$ be a set of pure states on $M_d$.
Then the states in $S$ possess  a common cloner if and only if the associated state vectors are orthogonal.
\end{theorem}
 
\begin{proof}
Sufficiency:
Extend the set of state vectors of states in $S$ to an orthonormal basis $\e_1,\ldots,\e_d$ of $\H:=\CC^d$,
and define an isometry $v:\H\to\H\ten\H$ by $v\e_i:=\e_i\ten\e_i$.
Then
   $$C:M_d\ten M_d\to M_d:\quad a\ten b\mapsto v^*(a\ten b)v$$
is a common cloner for all states in $S$, since, for $\ph=\inp{\e_j}{\cdot\e_j}$,
   $$\ph\circ C(a\ten b)=\inp{\e_j}{v^*(a\ten b)v\e_j}=\inp{\e_j\ten\e_j}{(a\ten b)\e_j\ten\e_j}=\ph\ten\ph(a\ten b)\;.$$
To prove necessity, suppose that $C$ is a cloner on $M_d$ for $S$.
By Stinespring's Theorem (Proposition \ref{PropStinespring}) there exist $k\in\NN$  and
an isometry $v:\H\to\H\ten\H\ten\CC^k$ such that
\begin{equation}\label{EqAppStine}
C(a\ten b)=v^*(a\ten b\ten\one_{M_k})v\;.
\end{equation}
For $\ph=\inp{\psi}{\cdot\psi}\in S$ we have, putting $p:=\ketbra{\psi}$:
\begin{eqnarray*}
1&=&\ph(p)^2=\ph\ten\ph(p\ten p)=\ph\circ C(p\ten p)\\
&=&\ph\bigl(v^*(p\ten p\ten\one_{M_k})v\bigr)=\inp{v\psi}{(p\ten p\ten\one_{M_k})v\psi}\;.
\end{eqnarray*}
It follows that $v\psi\in\psi\ten\psi\ten\CC^k$, say $v\psi=\psi\ten\psi\ten\a$ for some unit vector $\a\in\CC^k$.
Then, if $\inp{\th}{\cdot\th}$ is a another state in $S$, say with $v\th=\th\ten\th\ten\b$, we have
   $$\inp\psi\th=\inp{v\psi}{v\th}=\inp{\psi\ten\psi\ten\a}{\th\ten\th\ten\b}=\inp{\psi}{\th}^2\cdot\inp\a\b\;.$$
So either $\inp\psi\th=0$ or $\inp\psi\th\cdot\inp\a\b=1$.
As the states $\inp\psi{\cdot\psi}$ and $\inp\th{\cdot\th}$ are different, we have $|\inp\psi\th|<1$.
But $|\inp\a\b|\le1$. Therefore $\psi\perp\th$. 
\end{proof}

\noindent
{\bf Remark.}
The usual proofs of the ``no-cloning'' principle content themselves with the argument that there is no unitary $u$ with
    $$u:\quad\psi\ten\e\mapsto\psi\ten\psi$$
for $\psi\in\{\psi_1,\psi_2\}\subset\CC^d$ and a fixed ``register'' vector $\e\in\CC^d$,  
unless $\psi_1$ and $\psi_2$ are parallel or orthogonal.
It leaves one wondering why only operations on such small, $d^2$-dimensional, closed systems are considered as potential copiers.
Even \cite{NielsenChuang} (p. 532) yield to the temptation of simplicity,
although they admit that more sophisticated arguments are needed and can be given.
The importance of complete positivity as a minimal requirement was first realized by \cite{BarnumEtAl1},
M. Caves, C. Fuchs, R. Jozsa, and B. Shumacher. 


\bibliography{parthabib}

\end{document}